\documentclass[copyright]{eptcs}
\providecommand{\event}{QPL 2022}
\providecommand{\publicationstatus}{}

\usepackage{iftex}

\ifpdf
  \usepackage[T1]{fontenc}        % Recommended with pdflatex
\else
  \usepackage{breakurl}           % Not needed if you use pdflatex only.
\fi

\input{cliffordt2.sty}

\title{Generators and Relations for $2$-Qubit Clifford+$T$ Operators}
\author{Xiaoning Bian and Peter Selinger
  \institute{Dalhousie University}}

\providecommand{\titlerunning}{Generators and Relations for $2$-Qubit Clifford+$T$ Operators}
\providecommand{\authorrunning}{Xiaoning Bian and Peter Selinger}

\begin{document}
\maketitle

\begin{abstract}
  We give a presentation by generators and relations of the group of
  Clifford+$T$ operators on two qubits. The proof relies on an
  application of the Reidemeister-Schreier theorem to an earlier
  result of Greylyn, and has been formally verified in the proof
  assistant Agda.
\end{abstract}

%----------------------------------------------------------------------
\section{Introduction}

The simplification of Clifford+$T$ circuits is a topic of current
interest in quantum computing {\cite{AM-2019,debeaudrap2020fast,
    de_Beaudrap_2020, HC-2018, Kissinger-Wetering, Nam_2018}}. The
Clifford+$T$ gate set is both universal {\cite{Nielsen-Chuang}} and
convenient for quantum error correction {\cite{Buhrman2006}}, and is
therefore the preferred gate set for fault-tolerant quantum
computing. Generally, in a fault-tolerant regime, applying a Clifford
gate is some orders of magnitude cheaper than applying a $T$-gate, and
therefore, it is sensible to try to simplify circuits so as to
minimize the $T$-count {\cite{AMMR12}}. Many methods for doing so have
been proposed in the recent literature, including methods based on
matroid partitioning {\cite{Amy-Matroid}}, Reed-Muller codes
{\cite{AM-2019}}, and ZX calculus
{\cite{debeaudrap2020fast,de_Beaudrap_2020,Kissinger-Wetering}}. Regardless
of which method is used, the objective is to replace a Clifford+$T$
circuit by a simpler, but equivalent circuit. This requires being able
to tell when two circuits are equivalent. Surprisingly, no complete
set of relations for ancilla-free Clifford+$T$ circuits is currently
known, i.e., there is no known set of relations by which any two
equivalent Clifford+$T$ circuits can be transformed into each other.

In this paper, we give such a complete set of relations for the case
of $2$-qubit Clifford+$T$ circuits. We do this in several steps.
First, a presentation of the group $U_4(\Z[\frac{1}{\sqrt{2}},i])$ of
all unitary $4\times 4$-matrices over the ring
$\Z[\frac{1}{\sqrt{2}},i]$ is known due to the work of Greylyn
{\cite{greylyn2014generators}}.  Second, it is known that the group of
$2$-qubit Clifford+$T$ circuits is exactly the subgroup of this group
consisting of matrices whose determinant is in $\s{\pm 1, \pm i}$
{\cite{GilesSelinger2013}}.  Third, there is a theorem in group theory
called the Reidemeister-Schreier theorem, by which a complete set of
relations for a subgroup can be derived from a complete set of
relations for the supergroup. Fourth, since the resulting relations
are very long and complicated, we simplify them.

The last two steps of this procedure (applying the
Reidemeister-Schreier theorem and simplifying the resulting relations)
require a large amount of algebraic manipulations. Our longest
equational proof has 480 steps, each of which in turns requires a
lemma or rewrite procedure whose proof itself requires many equational
steps. Such proofs would be impossible to verify by hand, and even
verifying them by software is error-prone since it is hard to
guarantee that no unwarranted assumptions were used. For this reason,
we encoded our proof in machine-checkable form, using the proof
assistant Agda {\cite{Agda}}.

The rest of this paper is organized as follows. In
Section~\ref{sec:statement}, we state our main
result. Section~\ref{sec:outline} gives a brief overview of the
proof. In Section~\ref{sec:background}, we present the required
background material, including Greylyn's presentation of
$U_4(\Z[\frac{1}{\sqrt{2}},i])$, the Reidemeister-Schreier theorem,
and the Pauli rotation representation, which is an important tool for
manipulating Clifford+$T$ circuits. We also briefly describe our
reasons for formalizing our proof in a proof
assistant. Section~\ref{sec:proof} describes our formal proof of the
main result. In Section~\ref{sec:axioms}, we briefly discuss the
meaning of the Clifford+$T$ relations, and especially of the three
``non-obvious'' relations. Section~\ref{sec:conclusion} contains some
concluding remarks and ideas for future work.

% ----------------------------------------------------------------------
\section{Statement of the main result}
\label{sec:statement}

Recall that the set of Clifford operators is generated by the operators
\begin{equation*}\label{eqn-generators}
\omega = e^{i\pi/4},\quad
H = \frac{1}{\sqrt{2}}\zmatrix{cc}{1&1\\1&-1},\quad
S = \zmatrix{cc}{1&0\\0&i},\quad
\displaystyle \CZ =
\zmatrix{cccc}{1&0&0&0\\0&1&0&0\\0&0&1&0\\0&0&0&-1},
\end{equation*}
and is closed under multiplication and tensor product. Every such
operator $U$ is of size $2^n\times 2^n$ for some natural number $n$, and as
usual, we say that $U$ is an operator on $n$ qubits. We write
$\Clifford(n)$ for the group of $n$-qubit Clifford operators. It is
well-known that this group is finite for any given $n$ {\cite{selinger2013generators}}, and
therefore not universal for quantum computing. We obtain a universal 
gate set by also adding the $T$-gate as a generator.
\begin{equation*}
T = \zmatrix{cc}{1&0\\0&\omega},\quad
\end{equation*}
The resulting operators are called the Clifford+$T$ operators, and we write
$\CliffordT(n)$ for the $n$-qubit Clif\-ford+$T$ group.

In this paper, we focus on the case $n=2$. Our goal is to give a
complete presentation of the $2$-qubit Clifford+$T$ group in terms of
generators and relations. To ensure that all of our generators are
$4\times 4$-matrices, we introduce the following notation: we write
$T_0=T\otimes I$ and $T_1=I\otimes T$, and similarly for $H_0$, $H_1$,
$S_0$, and $S_1$.  We also identify the scalar $\omega$ with the
$4\times 4$-matrix $\omega I$.  Our main result is the following:

\begin{theorem}\label{thm:main}
  The $2$-qubit Clifford+$T$ group is presented by $(\Xx,\Ss)$, where
  the set of generators is
  \[
    \Xx=\s{\omega,H_0,H_1,S_0,S_1,T_0,T_1,\CZ},
  \]
  and the set of
  relations $\Ss$ is shown in Figure~\ref{fig-relations}.
\end{theorem}
%......................................................................
\begin{figure}
  % change the equation numbering, just for this Figure
  \newcounter{mytmpcounter}
  \setcounter{mytmpcounter}{\value{equation}}
  \setcounter{equation}{0}
  \makeatletter
  \renewcommand{\theequation}{C\@arabic\c@equation}
  \makeatother
  \def\scale{0.45}
  \def\spacing{-1ex}
  \def\xpacing{\vspace{-1ex}}
  
  (a) Monoidal relations:\xpacing
  \begin{eqnarray}
    \omega A &=& A \omega, \quad \mbox{where $A\in\s{H_i,S_i,T_i,\CZ}$}\label{rel-c1}\\
    A_0B_1 &=& B_1A_0,\quad \mbox{where $A,B\in\s{H,S,T}$}
  \end{eqnarray}

  (b) Order of Clifford group elements:\xpacing
  \begin{eqnarray}
    \omega^8 &=& \epsilon\\
    H_i^2 &=& \epsilon\\
    S_i^4 &=& \epsilon\\
    (S_iH_i)^3 &=& \omega \\
    \CZ^2 &=& \epsilon
  \end{eqnarray}
 
  (c) Remaining Clifford relations:\xpacing
  \begin{eqnarray}
    \m{\begin{qcircuit}[scale=\scale]
        \grid{3.5}{0,1}
        \gate{$S$}{1.25,1}
        \controlled{\dotgate}{2.5,0}{1}
      \end{qcircuit}
    } 
    &=& 
    \m{\begin{qcircuit}[scale=\scale]
        \grid{3.5}{0,1}
        \controlled{\dotgate}{1,0}{1}
        \gate{$S$}{2.25,1}
      \end{qcircuit}
    }\label{rel-c8}
    \\\nonumber\\[\spacing]
    \m{\begin{qcircuit}[scale=\scale]
        \grid{3.5}{0,1}
        \gate{$S$}{1.25,0}
        \controlled{\dotgate}{2.5,0}{1}
      \end{qcircuit}
    } 
    &=& 
    \m{\begin{qcircuit}[scale=\scale]
        \grid{3.5}{0,1}
        \controlled{\dotgate}{1,0}{1}
        \gate{$S$}{2.25,0}
      \end{qcircuit}
    }
    \\\nonumber\\[\spacing]
    \m{\begin{qcircuit}[scale=\scale]
        \grid{8}{0,1}
        \gate{$H$}{1.25,1}
        \gate{$S$}{2.75,1}
        \gate{$S$}{4.25,1}
        \gate{$H$}{5.75,1}
        \controlled{\dotgate}{7,0}{1}
        \invisiblegate{1,0}
      \end{qcircuit}
    } 
    &=& 
    \m{\begin{qcircuit}[scale=\scale]
        \grid{8}{0,1}
        \controlled{\dotgate}{1,0}{1}
        \gate{$S$}{2.25,0}
        \gate{$S$}{3.75,0}
        \gate{$H$}{2.25,1}
        \gate{$S$}{3.75,1}
        \gate{$S$}{5.25,1}
        \gate{$H$}{6.75,1}
      \end{qcircuit}
    }
    \\\nonumber\\[\spacing]
    \m{\begin{qcircuit}[scale=\scale]
        \grid{8}{0,1}
        \gate{$H$}{1.25,0}
        \gate{$S$}{2.75,0}
        \gate{$S$}{4.25,0}
        \gate{$H$}{5.75,0}
        \controlled{\dotgate}{7,0}{1}
        \invisiblegate{1,1}
      \end{qcircuit}
    } 
    &=& 
    \m{\begin{qcircuit}[scale=\scale]
        \grid{8}{0,1}
        \controlled{\dotgate}{1,0}{1}
        \gate{$S$}{2.25,1}
        \gate{$S$}{3.75,1}
        \gate{$H$}{2.25,0}
        \gate{$S$}{3.75,0}
        \gate{$S$}{5.25,0}
        \gate{$H$}{6.75,0}
      \end{qcircuit}
    }
    \\\nonumber\\[\spacing]
    \m{\begin{qcircuit}[scale=\scale]
        \grid{4.5}{0,1}
        \controlled{\dotgate}{1,0}{1}
        \gate{$H$}{2.25,1}
        \controlled{\dotgate}{3.5,0}{1}
        \invisiblegate{1,0}
      \end{qcircuit}
    }
    &=& 
    \m{\begin{qcircuit}[scale=\scale]
        \gridx{1.5}{11}{0,1}
        \gate{$S$}{2.75,1}
        \gate{$H$}{4.25,1}
        \controlled{\dotgate}{5.5,0}{1}
        \gate{$S$}{6.75,0}
        \gate{$S$}{6.75,1}
        \gate{$H$}{8.25,1}
        \gate{$S$}{9.75,1}
      \end{qcircuit}
    }\cdot\omega\inv
    \\\nonumber\\[\spacing]
    \m{\begin{qcircuit}[scale=\scale]
        \grid{4.5}{0,1}
        \controlled{\dotgate}{1,0}{1}
        \gate{$H$}{2.25,0}
        \controlled{\dotgate}{3.5,0}{1}
        \invisiblegate{1,1}
      \end{qcircuit}
    }
    &=& 
    \m{\begin{qcircuit}[scale=\scale]
        \gridx{1.5}{11}{0,1}
        \gate{$S$}{2.75,0}
        \gate{$H$}{4.25,0}
        \controlled{\dotgate}{5.5,0}{1}
        \gate{$S$}{6.75,1}
        \gate{$S$}{6.75,0}
        \gate{$H$}{8.25,0}
        \gate{$S$}{9.75,0}
      \end{qcircuit}
    }\cdot\omega\inv
  \end{eqnarray}

  (d) ``Obvious'' relations involving $T$:\xpacing
  \begin{eqnarray}
    T_i^2 &=& S_i \label{rel-c14}\\
    (T_iH_iS_iS_iH_i)^2 &=& \omega\label{rel-c15}
    \\\nonumber\\[\spacing]
    \m{\begin{qcircuit}[scale=\scale]
        \grid{3.5}{0,1}
        \gate{$T$}{1.25,1}
        \controlled{\dotgate}{2.5,0}{1}
      \end{qcircuit}
    } 
    &=& 
    \m{\begin{qcircuit}[scale=\scale]
        \grid{3.5}{0,1}
        \controlled{\dotgate}{1,0}{1}
        \gate{$T$}{2.25,1}
      \end{qcircuit}
    }
    \label{rel-c16}
    \\\nonumber\\[\spacing]
    \m{\begin{qcircuit}[scale=\scale]
        \grid{9}{0,1}
        \gate{$H$}{1.25,0}
        \controlled{\dotgate}{2.5,0}{1}
        \gate{$H$}{3.75,0}
        \gate{$H$}{3.75,1}
        \controlled{\dotgate}{5.0,0}{1}
        \gate{$H$}{6.25,1}
        \gate{$T$}{7.75,1}
      \end{qcircuit}
    }
    &=&
    \m{\begin{qcircuit}[scale=\scale]
        \gridx{-1.5}{7.5}{0,1}
        \gate{$T$}{-0.25,0}
        \gate{$H$}{1.25,0}
        \controlled{\dotgate}{2.5,0}{1}
        \gate{$H$}{3.75,0}
        \gate{$H$}{3.75,1}
        \controlled{\dotgate}{5.0,0}{1}
        \gate{$H$}{6.25,1}
      \end{qcircuit}
    }
    \qquad
    \label{rel-c17}
  \end{eqnarray}

  (e) ``Non-obvious'' relations involving $T$:
  \begin{eqnarray}
    \scalebox{1}{\m{
        \begin{qcircuit}[scale=\scale]
          \grid{12.00}{0,1}
          \controlled{\notgate}{1,0}{1};
          \gate{$T$}{2.25,0};
          \gate{$H$}{3.75,0};
          \gate{$T\da$}{5.25,0};
          \whitecontrolled{\notgate}{6.5,0}{1};
          \gate{$T$}{7.75,0};
          \gate{$H$}{9.25,0};
          \gate{$T\da$}{10.75,0};
        \end{qcircuit}
    }}
    &=&
    \scalebox{1}{\m{
        \begin{qcircuit}[scale=\scale]
          \grid{12.00}{0,1}
          \gate{$T$}{1.25,0};
          \gate{$H$}{2.75,0};
          \gate{$T\da$}{4.25,0};
          \whitecontrolled{\notgate}{5.50,0}{1};
          \gate{$T$}{6.75,0};
          \gate{$H$}{8.25,0};
          \gate{$T\da$}{9.75,0};
          \controlled{\notgate}{11,0}{1};
        \end{qcircuit}
    }}
    \label{eqn-rela}
    \\\nonumber\\[\spacing]
    \scalebox{1}{\m{
 	\begin{qcircuit}[scale=\scale]
          \grid{18.00}{0,1}
          \controlled{\notgate}{1.00,0}{1};
          \gate{$T$}{2.25,0};
          \gate{$H$}{3.75,0};
          \gate{$T$}{5.25,0};
          \gate{$H$}{6.75,0};
          \gate{$T\da$}{8.25,0};
          \whitecontrolled{\notgate}{9.50,0}{1};
          \gate{$T$}{10.75,0};
          \gate{$H$}{12.25,0};
          \gate{$T^\dagger$}{13.75,0};
          \gate{$H$}{15.25,0};
          \gate{$T^\dagger$}{16.75,0};
        \end{qcircuit}
    }}
    &=&
    \scalebox{1}{\m{
 	\begin{qcircuit}[scale=\scale]
          \grid{18.00}{0,1}
          \gate{$T$}{1.25,0};
          \gate{$H$}{2.75,0};
          \gate{$T$}{4.25,0};
          \gate{$H$}{5.75,0};
          \gate{$T\da$}{7.25,0};
          \whitecontrolled{\notgate}{8.5,0}{1};
          \gate{$T$}{9.75,0};
          \gate{$H$}{11.25,0};
          \gate{$T^\dagger$}{12.75,0};
          \gate{$H$}{14.25,0};
          \gate{$T^\dagger$}{15.75,0};
          \controlled{\notgate}{17,0}{1};
        \end{qcircuit}
    }}
    \qquad
    \label{eqn-relb}
    \\\nonumber\\[\spacing]
    \scalebox{1}{\m{
        \begin{qcircuit}[scale=\scale]
          \grid{10.00}{0,1}
          \whitecontrolled{\gate{$H$}}{1.25,0}{1};
          \gate{$T$}{2.75,0};
          \controlled{\gate{$H$}}{4.25,0}{1};
          \whitecontrolled{\gate{$H$}}{5.75,1}{0};
          \gate{$T$}{7.25,1};
          \controlled{\gate{$H$}}{8.75,1}{0};
        \end{qcircuit}
    }} &=& \scalebox{1}{\m{
        \begin{qcircuit}[scale=\scale]
          \grid{10.00}{0,1}
          \whitecontrolled{\gate{$H$}}{1.25,1}{0};
          \gate{$T$}{2.75,1};
          \controlled{\gate{$H$}}{4.25,1}{0};
          \whitecontrolled{\gate{$H$}}{5.75,0}{1};
          \gate{$T$}{7.25,0};
          \controlled{\gate{$H$}}{8.75,0}{1};
        \end{qcircuit}
    }}\label{eqn-relc}\label{rel-c20}
  \end{eqnarray}
  \caption{Relations for 2-qubit Clifford+$T$ operators. Here
    $i\in\s{0,1}$.}\label{fig-relations}

  % restore the original equation numbering
  \setcounter{equation}{\value{mytmpcounter}}
\end{figure}
%......................................................................

In Figure~\ref{fig-relations}, we have used circuit notation to
express some of the relations; for example, we have written
\[
\m{\begin{qcircuit}[scale=0.5]
    \grid{2.5}{0,1}
    \invisiblegate{1.25,0}
    \gate{$T$}{1.25,1}
  \end{qcircuit}
},
\quad
\m{\begin{qcircuit}[scale=0.5]
    \grid{2.5}{0,1}
    \invisiblegate{1.25,1}
    \gate{$T$}{1.25,0}
  \end{qcircuit}
},
\quad\mbox{and}\quad
\m{\begin{qcircuit}[scale=0.5]
    \grid{2}{0,1}
    \controlled{\dotgate}{1,0}{1}
  \end{qcircuit}
}
\]
for $T_0$, $T_1$, and $\CZ$, respectively. Note that the qubits are
numbered from top to bottom. We write circuits in the same order as
matrix multiplication. Moreover, in relations
{\eqref{eqn-rela}--\eqref{eqn-relc}}, we have used a number of
abbreviations; these are defined in Figure~\ref{fig-abbreviations}.
The empty word is denoted $\epsilon$.

%......................................................................
\begin{figure}
  \[
  \begin{array}{rcl}
    T\da &=& T^7
    \\
    S\da &=& S^3
    \\\nonumber\\[0ex]
    \scalebox{1}{\m{
        \begin{qcircuit}[scale=0.5]
          \grid{2}{0,1}
          \controlled{\notgate}{1,0}{1};
          \invisiblegate{1,0}
        \end{qcircuit}
    }} &=& \scalebox{1}{\m{
        \begin{qcircuit}[scale=0.5]
          \grid{5.00}{0,1}
          \gate{$H$}{1.25,0};
          \controlled{\dotgate}{2.5,1}{0};
          \gate{$H$}{3.75,0};
        \end{qcircuit}
    }}
    \\\nonumber\\[0ex]
    \scalebox{1}{\m{
        \begin{qcircuit}[scale=0.5]
          \grid{2}{0,1}
          \whitecontrolled{\notgate}{1,0}{1};
          \invisiblegate{1,0}
        \end{qcircuit}
    }} &=& \scalebox{1}{\m{
        \begin{qcircuit}[scale=0.5]
          \grid{8.00}{0,1}
          \controlled{\notgate}{1,0}{1};
          \gate{$H$}{2.25,0};
          \gate{$S$}{3.75,0};
          \gate{$S$}{5.25,0};
          \gate{$H$}{6.75,0};
        \end{qcircuit}
    }}
    \\\nonumber\\[0ex]
    \scalebox{1}{\m{
        \begin{qcircuit}[scale=0.5]
          \grid{2}{0,1}
          \controlled{\gate{$H$}}{1,0}{1};
          \invisiblegate{1,0}
        \end{qcircuit}
    }} &=& \scalebox{1}{\m{
        \begin{qcircuit}[scale=0.5]
          \grid{11.25}{0,1}
          \gate{$S$}{1.25,0};
          \gate{$H$}{2.75,0};
          \gate{$T$}{4.25,0};
          \controlled{\notgate}{5.5,0}{1};
          \gate{$T\da$}{7,0};
          \gate{$H$}{8.5,0};
          \gate{$S\da$}{10,0};
        \end{qcircuit}
    }}
    \\\nonumber\\[0ex]
    \scalebox{1}{\m{
        \begin{qcircuit}[scale=0.5]
          \grid{2}{0,1}
          \whitecontrolled{\gate{$H$}}{1,0}{1};
          \invisiblegate{1,0}
        \end{qcircuit}
    }} &=& \scalebox{1}{\m{
        \begin{qcircuit}[scale=0.5]
          \grid{4.0}{0,1}
          \controlled{\gate{$H$}}{1.25,0}{1};
          \gate{$H$}{2.75,0};
        \end{qcircuit}
    }}
  \end{array}
  \]
  \caption{Abbreviations used in circuit notations}
  \label{fig-abbreviations}
\end{figure}
%......................................................................

% ----------------------------------------------------------------------
\section{Proof outline}\label{sec:outline}

In a nutshell, the proof can be described in a few sentences. It
proceeds as follows. Let $R=\Z[\frac{1}{\sqrt{2}},i]$ be the smallest
subring of the complex numbers containing $\frac{1}{\sqrt{2}}$ and
$i$, and let $G=U_4(R)$ be the group of unitary $4\times 4$-matrices
with entries in $R$. Then it is clear that $\CliffordT(2)$ is a
subgroup of $G$, because all of its generators belong to
$G$. Moreover, from {\cite{GilesSelinger2013}}, it is known that
$\CliffordT(2)$ is precisely equal to the subgroup of $G$ consisting
of matrices whose determinant is a power of $i$.  A presentation of
$G$ by generators and relations was given by Greylyn
{\cite{greylyn2014generators}}. There is a general procedure, called
the Reidemeister-Schreier procedure
{\cite{Reidemeister1927,Schreier1927}}, for finding generators and
relations of a subgroup, given generators and relations of the
supergroup. Applying this procedure therefore yields a complete set of
relations for $\CliffordT(2)$.

While in principle, the above proof outline suffices to prove
Theorem~\ref{thm:main}, in practice there is a large amount of
non-trivial work involved in generating and simplifying the actual
relations. For this reason, we have formalized Theorem~\ref{thm:main}
and its proof in the proof assistant Agda. This allows the proof to be
independently checked without too much manual work.

% ----------------------------------------------------------------------
\section{Background}
\label{sec:background}

\subsection{Presentation of \texorpdfstring{$U_4(\Z[\frac{1}{\sqrt{2}},i])$}{U₄(ℤ[½, i])}}

As usual, $\Z$ is the ring of integers. Let
$R=\Z[\frac{1}{\sqrt{2}},i]$ be the smallest subring of the complex
numbers containing $\frac{1}{\sqrt{2}}$ and $i$.  Let
$\omega=e^{i\pi/4}$ be an 8th root of unity, and note that
$\omega = \frac{1+i}{\sqrt{2}}\in R$.  As before, $U_4(R)$ is the
group of unitary $4\times 4$-matrices with entries in $R$.

Greylyn {\cite{greylyn2014generators}} gave a presentation of $U_4(R)$
by generators and relations. His generators are $\omega_{[j]}$,
$X_{[j,k]}$, and $H_{[j,k]}$, where $j,k\in\s{0,...,3}$ and $j<k$.
The relations are shown in Figure~\ref{fig:greylyn-relations}.  The
intended interpretation of the generators is as 1- and 2-level
matrices; specifically, $\omega_{[j]}$ is like the identity matrix,
except with $\omega$ in the $j$th row and column, and $X_{[j,k]}$ and
$H_{[j,k]}$ are like identity matrices, except with the entries of
$X$, respectively $H$, in the $j$th and $k$th rows and columns, like
this:
\[
\def\scale{0.95}
\def\xskip{\hspace{0.3em}}  
\omega_{[\jay]} \xskip=\xskip
\scalebox{\scale}{$\kbordermatrix{
  & \cdots & \jay & \cdots \\
  \svdots & I & 0 & 0 \\
  \jay & 0 & \omega & 0  \\
  \svdots & 0 & 0 & I \\
}$},
\quad
X_{[j,k]}\xskip=\xskip
\scalebox{\scale}{$\kbordermatrix{
  &   \dots & j & \dots & k & \dots \\
  \svdots &   I & 0 & 0 & 0 & 0 \\
  j &   0 & 0 & 0 & 1 & 0 \\
  \svdots &   0 & 0 & I & 0 & 0 \\
  k &   0 & 1 & 0 & 0 & 0 \\
  \svdots &   0 & 0 & 0 & 0 & I
}$},
\quad
H_{[j,k]}\xskip=\xskip
\scalebox{\scale}{$\kbordermatrix{
  &   \dots & j & \dots & k & \dots \\
  \svdots &   I & 0 & 0 & 0 & 0 \\
  j &   0 & \frac{1}{\sqrt{2}} & 0 & \frac{1}{\sqrt{2}} & 0 \\
  \svdots &   0 & 0 & I & 0 & 0 \\
  k &   0 & \frac{1}{\sqrt{2}} & 0 & -\frac{1}{\sqrt{2}} & 0 \\
  \svdots &   0 & 0 & 0 & 0 & I
}$}.
\]
Note that we index rows and columns of matrices starting from 0,
whereas Greylyn indexed them starting from 1. Greylyn's result is the
following:

\begin{theorem}[Greylyn \cite{greylyn2014generators}]\label{thm:greylyn}
  A presentation of the group $U_4(R)$ is given by $(\Yy,\Rr)$,
  where the set of generators is
  $\Yy=\s{\omega_{[j]}, X_{[j,k]}, H_{[j,k]} \mid
    \mbox{$j,k\in\s{1,...,4}$ and $j<k$}}$, and the set of relations
  $\Rr$ is shown in Figure~\ref{fig:greylyn-relations}.
\end{theorem}
  
% ......................................................................
\begin{figure}
  % change the equation numbering, just for this Figure
  \setcounter{mytmpcounter}{\value{equation}}
  \setcounter{equation}{0}
  \makeatletter
  \renewcommand{\theequation}{G\@arabic\c@equation}
  \makeatother
  \def\b#1{\makebox[12ex][l]{$#1$}}
  
  (a) Order of generators:
  \begin{eqnarray}
    \omega^8_{[\jay]}&=&\epsilon\\
    H^2_{[\jay,\kay]}&=&\epsilon\\
    X^2_{[\jay,\kay]}&=&\epsilon
  \end{eqnarray}

  (b) Disjoint generators commute:
  \begin{eqnarray}
    \omega_{[\jay]}\omega_{[\kay]}&=&\b{\omega_{[\kay]}\omega_{[\jay]},} \quad \mbox{where $\jay\neq\kay$}\\
    \omega_{[\ell]}H_{[\jay,\kay]}&=&\b{H_{[\jay,\kay]}\omega_{[\ell]},} \quad \mbox{where $\ell\neq\jay,\kay$}\\
    \omega_{[\ell]}X_{[\jay,\kay]}&=&\b{X_{[\jay,\kay]}\omega_{[\ell]},} \quad \mbox{where $\ell\neq\jay,\kay$}\\
    H_{[\jay,\kay]}H_{[\ell,\tee]}&=&\b{H_{[\ell,\tee]}H_{[\jay,\kay]},} \quad \mbox{where $\{\ell,\tee\}\cap\{\jay,\kay\}=\emptyset$}\\
    H_{[\jay,\kay]}X_{[\ell,\tee]}&=&\b{X_{[\ell,\tee]}H_{[\jay,\kay]},} \quad \mbox{where $\{\ell,\tee\}\cap\{\jay,\kay\}=\emptyset$}\\
    X_{[\jay,\kay]}X_{[\ell,\tee]}&=&\b{X_{[\ell,\tee]}X_{[\jay,\kay]},} \quad \mbox{where $\{\ell,\tee\}\cap\{\jay,\kay\}=\emptyset$}
  \end{eqnarray}

  (c) $X$ permutes indices:
  \begin{eqnarray}
    X_{[\jay,\kay]}\omega_{[\kay]}&=&\omega_{[\jay]}X_{[\jay,\kay]}\\
    X_{[\jay,\kay]}\omega_{[\jay]}&=&\omega_{[\kay]}X_{[\jay,\kay]}\\
    X_{[\jay,\kay]}X_{[\jay,\ell]}&=&X_{[\kay,\ell]}X_{[\jay,\kay]}\\
    X_{[\jay,\kay]}X_{[\ell,\jay]}&=&X_{[\ell,\kay]}X_{[\jay,\kay]}\\
    X_{[\jay,\kay]}H_{[\jay,\ell]}&=&H_{[\kay,\ell]}X_{[\jay,\kay]}\\
    X_{[\jay,\kay]}H_{[\ell,\jay]}&=&H_{[\ell,\kay]}X_{[\jay,\kay]}
  \end{eqnarray}

  (d) $\omega_{[\jay]}\omega_{[\kay]}$ is diagonal:
  \begin{eqnarray}
    \omega_{[\jay]}\omega_{[\kay]}X_{[\jay,\kay]}&=&X_{[\jay,\kay]}\omega_{[\jay]}\omega_{[\kay]}\\
    \omega_{[\jay]}\omega_{[\kay]}H_{[\jay,\kay]}&=&H_{[\jay,\kay]}\omega_{[\jay]}\omega_{[\kay]}
  \end{eqnarray}

  (e) Relations for $H$:
  \begin{eqnarray}
    H_{[\jay,\kay]}X_{[\jay,\kay]}&=&\omega^4_{[k]}H_{[\jay,\kay]}\\
    H_{[\jay,\kay]}\omega^2_{[\jay]}H_{[\jay,\kay]}&=& \omega^6_{[\jay]}H_{[\jay,\kay]} \omega^3_{[\jay]}\omega^5_{[\kay]}\\
    H_{[\jay,\kay]}H_{[\ell,\tee]} H_{[\jay,\ell]}H_{[\kay,\tee]}&= & H_{[\jay,\ell]}H_{[\kay,\tee]}H_{[\jay,\kay]}H_{[\ell,\tee]}, \quad \mbox{where $\kay<\ell$}
  \end{eqnarray}

  \caption{Greylyn's relations for
    $U_4(\Z[\frac{1}{\sqrt{2}},i])$. Whenever we use a generator
    $X_{[\jay,\kay]}$ or $H_{[\jay,\kay]}$, we implicitly assume
    that $j<k$.}
  \label{fig:greylyn-relations}

  % restore the original equation numbering
  \setcounter{equation}{\value{mytmpcounter}}
\end{figure}

% ......................................................................

\subsection{The Reidemeister-Schreier theorem for monoids}
\label{ssec:reidemeister-schreier}

The Reidemeister-Schreier theorem is a theorem in group theory that
allows one to derive a complete set of relations for a subgroup from a
complete set of relations for the supergroup, given enough information
about the cosets. We will use a version of the Reidemeister-Schreier
theorem that works for monoids, which we now describe. To our
knowledge, this monoid formulation of the Reidemeister-Schreier
theorem does not appear in the literature.

If $X$ is a set, let us write $X^*$ for the set of finite sequences of
elements of $X$, which we also call {\em words} over the alphabet $X$.
We write $w\cdot v$ or simply $wv$ for the concatenation of words,
making $X^*$ into a monoid. The unit of this monoid is the empty word
$\epsilon$. As usual, we identify $X$ with the set of one-letter
words.

Let $G$ be a monoid and let $X\seq G$ be a subset of $G$. We write
$\gen{X}$ for the smallest submonoid of $G$ containing $X$, and we say
that $X$ {\em generates} $G$ if $\gen{X}=G$. Given any word $w\in
X^*$, we write $\semG{w}\in G$ for the canonical interpretation of $w$
in $G$, i.e., $\semG{-}:X^*\to G$ is the unique monoid homomorphism
such that $\semG{x}=x$ for all $x\in X$.

A {\em relation} over $X$ is an element of $X^*\times X^*$, i.e., an
ordered pair of words. We say that a relation $(w,v)$ is {\em valid}
in $G$ if $\semG{w}=\semG{v}$. If $\Ss$ is a set of relations over
$X$, we write $\simS$ for the smallest congruence relation on $X^*$
containing $\Ss$.  Here, as usual, a congruence relation is an
equivalence relation that is compatible with the monoid operation,
i.e., such that $w\sim v$ and $w'\sim v'$ implies $ww'\sim vv'$. Given
a set $X$ of generators for a monoid $G$ and a set $\Ss$ of valid
relations, we say that $\Ss$ is {\em complete} if for all
$w,v\in X^*$, $\semG{w}=\semG{v}$ implies $w\simS v$. In that case, we
also say that $(X,\Ss)$ is a {\em presentation by generators and
  relations} (or simply {\em presentation}) of $G$.

\begin{definition}\label{def:f**}
  Given sets $X,Y$ and a function $f:X\to Y^*$, let $f^*:X^*\to Y^*$
  be the unique monoid homomorphism extending $f$. Concretely, $f^*$
  is given by $f^*(x_1\ldots x_n) = f(x_1)\cdot\ldots\cdot f(x_n)$.

  More generally, given sets $C,X,Y$ and a function $f:C\times X\to
  Y^*\times C$, let $f^{**} : C\times X^*\to Y^*\times C$ be the
  function defined by $f^{**}(c_0, x_1\ldots x_n) =
  (w_1\cdot\ldots\cdot w_n, c_n)$, where $f(c_{i-1}, x_i) = (w_i,
  c_i)$ for all $i=1,\ldots,n$.
\end{definition}

Note that in case $C$ is a singleton, the functions $f^{*}$ and
$f^{**}$ are essentially the same. In general, the difference is that
$f^{**}$ also keeps a ``state'' in the form of an element of $C$.

\begin{theorem}[Reidemeister-Schreier theorem for monoids]\label{thm:reide}
  Let $X$ and $Y$ be sets, and let $\Ss$ and $\Rr$ be sets of
  relations over $X$ and $Y$, respectively. Suppose that the
  following additional data is given:
  \begin{itemize}
  \item a set $C$ with a distinguished element $I\in C$,
  \item a function $f : X\to Y^*$,
  \item a function $h : C\times Y\to X^*\times C$,
  \end{itemize}
  subject to the following conditions:
  \begin{enumerate}\alphalabels
  \item For all $x\in X$, if $h^{**}(I, f(x)) = (v, c)$, then
    $v \simS x$ and $c=I$.
  \item For all $c\in C$ and $w,w'\in Y^*$ with $(w,w')\in\Rr$, if
    $h^{**}(c,w)=(v,c')$ and $h^{**}(c,w')=(v',c'')$ then
    $v \simS v'$ and $c'=c''$.
  \end{enumerate}
  Then for all $v,v'\in X^*$, $f^*(v)\simR f^*(v')$ implies $v\simS v'$.
\end{theorem}

To better understand the utility of this theorem, let us briefly
provide some context. First, we note that we will be using this
theorem in the case where $G$ is a monoid, $H$ is a submonoid of $G$,
$(Y,\Rr)$ is a presentation of $G$, $X$ is a set of generators for
$H$, and we wish to show that some proposed set of relations $\Ss$ is
complete for $H$. Assuming that all hypotheses of
Theorem~\ref{thm:reide} are satisfied, and further assuming that
$f$ represents the inclusion function of $H$ into $G$, i.e., that
for all $x\in X$, $\semG{f(x)} = \semH{x}$, the completeness of $\Ss$ then
follows. Namely, $\semH{v}=\semH{v'}$ implies
$\semG{f^*(v)}=\semG{f^*(v')}$, which implies $f^*(v)\simR f^*(v')$ by
completeness of $\Rr$, which implies $v \simS v'$ by
Theorem~\ref{thm:reide}.

To see how the theorem works, it is useful to further concentrate on
the case where $G$ and $H$ are groups, although the theorem itself
does not require this. In the case of groups, one would typically
consider the set $H\backslash G = \s{Hc\mid c\in G}$ of right cosets
of $H$ in $G$, and one would let $C$ be a set of chosen coset
representatives. The function $f$ is then chosen to assign to each
$x\in X$ some word $w\in Y^*$ such that $\semH{x}=\semG{w}$. The
function $h$ is chosen to assign to each pair of a coset
representative $c\in C$ and generator $y\in Y$ the unique coset
representative $c'\in C$ and some word $v\in X^*$ such that
$c\semG{y} = \semH{v}c'$. Conditions (a) and (b) are then sufficient
for the set of relations $\Ss$ to be complete. In the more general
case of monoids, $G$ is not necessarily partitioned into cosets, but
the method works anyway, provided that appropriate $C$, $f$, and $h$
can be chosen.

\begin{proof}[Proof of Theorem~\ref{thm:reide}]
  Let us say that a word $w\in Y^*$ is {\em special} if $h^{**}(I,w) =
  (v,I)$ for some $v\in X^*$. Let $\Ys^*$ be the set of special
  words. By definition of $h^{**}$, the empty word is special and
  special words are closed under concatenation, so $\Ys^*$ is a
  submonoid of $Y^*$. Moreover, the image of $f$ is special by
  property (a), and therefore the image of $f^*$ is also
  special. Finally, there is a translation back from special words in
  $Y$ to words in $X$: define $g:\Ys^*\to X^*$ by letting $g(w)=v$
  where $h^{**}(I,w) = (v,I)$. Clearly, $g$ is a monoid homomorphism.

  \smallskip
  \noindent
  Claim A: for all $v\in X^*$, we have $v\simS g(f^*(v))$.
  Proof: Since both $g$ and $f^*$ are monoid homomorphisms and $\simS$
  is a congruence, it suffices to show this in the case when $v\in X$
  is a generator. But in that case, it holds by assumption~(a).

  \smallskip
  \noindent
  Claim B: for all $w,w'\in Y^*$ and $c\in C$, if $w\simR w'$ and
  $h^{**}(c,w) = (v,d)$ and $h^{**}(c,w') = (v',d')$, then $v\simS v'$
  and $d=d'$. Proof: define a relation $\sim$ on $Y^*$ by $w\sim w'$
  if for all $c\in C$, $h^{**}(c,w) = (v,d)$ and
  $h^{**}(c,w') = (v',d')$ implies $v\simS v'$ and $d=d'$. We must
  show that $w\simR w'$ implies $w\sim w'$. Since $\simR$ is, by
  definition, the smallest congruence containing $\Rr$, it suffices to
  show that $\sim$ is a congruence containing $\Rr$. The fact that
  $\sim$ is reflexive, symmetric, and transitive is obvious from its
  definition. The fact that it is a congruence follows from the
  definition of $h^{**}$ and the fact that $\simS$ is a
  congruence. Finally, $\sim$ contains $\Rr$ by assumption~(b).
  
  Note that, as a special case of claim B, we also have the following:
  if $w,w'\in \Ys^*$ are special words, then $w\simR w'$ implies
  $g(w)\simS g(w')$. This follows directly from the definition of $g$.
  
  To finish the proof of the Reidemeister-Schreier theorem, let
  $v,v'\in X^*$ and assume that $f^*(v)\simR f^*(v')$. Then we have:
  \[
  v ~\simS~ g(f^*(v)) ~\simS~ g(f^*(v')) ~\simS~ v',
  \]
  where the first and last congruence holds by claim A, and the
  middle one holds by the special case of claim B. Therefore, $v\simS
  v'$ as claimed.
\end{proof}

\begin{corollary}
  Let $G$ be a monoid with presentation $(Y,\Rr)$, where $Y\seq G$.
  Suppose $H\seq G$ is a submonoid and $X$ is a set of generators for
  $H$. Let $\Ss$ be a set of valid relations for $H$. Assume a set $C$
  and functions $f$ and $h$ are given, satisfying the hypotheses of
  Theorem~\ref{thm:reide}, and assume that $f$ represents the
  inclusion function of $H$ into $G$, i.e., that $x\in X$,
  $\semG{f(x)} = \semH{x}$. Then $\Ss$ is a complete set of relations
  for $H$.\qed
\end{corollary}

% ----------------------------------------------------------------------
\subsection{Pauli rotation representation}
\label{ssec:pauli}

One of the problems we face in applying the Reidemeister-Schreier
theorem is that we must show that a large number of
(computer-generated) Clifford+$T$ relations follow from the relations
in Figure~\ref{fig-relations}. It would be very useful if this task
could be automated. Ideally, the relations in
Figure~\ref{fig-relations} could be turned into a set of rewrite rules
with the property that every Clifford+$T$ circuit can be rewritten to
a unique {\em normal form}; in that case, to show that a given
relation follows from the ones in Figure~\ref{fig-relations}, it would
be sufficient to reduce the left-hand and right-hand sides to normal
form and check that they are equal.

Unfortunately, no such rewrite system or normal form is known.
Instead, the best we can do is a semi-automated process in which words
are rewritten to something that is ``almost'' a normal form, i.e., not
quite unique, but close enough so that many relations can be proved
automatically, and the rest are more easily solvable by hand.

For this, the {\em Pauli rotation representation} of Clifford+$T$
operators turns out to be useful. This representation was first
described in {\cite[Section 3]{GKMR}}. We start by noting that the
$T$-gate is a linear combination of the identity $I$ and the Pauli
operator $Z$. Specifically:
\begin{equation}
T ~=~ \zmatrix{cc}{1&0\\0&\omega}
~=~ \frac{1+\omega}{2} I + \frac{1-\omega}{2} Z.
\end{equation}
Therefore, an operator $A$ commutes with $T$ if and only if it
commutes with $Z$. More generally, given any $n$-qubit Pauli operator
$P$, define
\begin{equation}\label{eqn:def-rsyll}
\Rsyll{P} ~=~ \frac{1+\omega}{2} I + \frac{1-\omega}{2} P.
\end{equation}
Note that $\Rsyll{Z}=T$.  We refer to the operators $\Rsyll{P}$ as
{\em (45 degree) Pauli rotations}. Note that $\Rsyll{P}$ is not a
Pauli operator; we call it a Pauli rotation because it is a rotation
about a Pauli axis. By {\eqref{eqn:def-rsyll}}, it is again obvious
that an operator $A$ commutes with $\Rsyll{P}$ if and only if it
commutes with $P$.  Moreover, from {\eqref{eqn:def-rsyll}}, we get the
following fundamental property of Pauli rotations:
\begin{equation}\label{eqn:conjugate}
CPC\inv = Q \quad\mbox{if and only if}\quad C\Rsyll{P}C\inv = \Rsyll{Q}.
\end{equation}
Let
$Z_{(i)} = I\otimes\ldots\otimes I\otimes Z\otimes I\otimes \ldots
\otimes I$ be the $n$-qubit Pauli operator with $Z$ acting on the
$i$th qubit, and similarly
$T_{(i)} = I\otimes\ldots\otimes I\otimes T\otimes I\otimes \ldots
\otimes I = \Rsyll{Z_{(i)}}$. Since the Clifford operators act
transitively on the set of non-trivial self-adjoint Pauli operators by
conjugation, for every such $n$-qubit Pauli operator $P$, there exists
a (non-unique) Clifford operator $C$ such that $CZ_{(1)}C\inv = P$,
and therefore $CT_{(1)}C\inv = \Rsyll{P}$. We therefore see that all
of the Pauli rotations are Clifford conjugates of the $T_{(1)}$-gate.

Next, we note that every Clifford+$T$ operator can be written as a
product of Pauli rotations followed by a single Clifford
operator. Specifically, by definition, every Clifford+$T$ operator can
be written as
\[ C_1 T_{(i_1)} C_2 T_{(i_2)} C_3 \cdots C_n T_{(i_{n})} C_{n+1}.
\]
For all $k$, let $D_k = C_1C_2\cdots C_k$, so that $C_k = D_{k-1}\inv
D_k$.  Then the above can be rewritten as
\begin{eqnarray*}
  C_1 T_{(i_1)} C_2 T_{(i_2)} C_3 \cdots C_n T_{(i_{n})} C_{n+1}
  &=& C_1 \Rsyll{Z_{(i_1)}} C_2 \Rsyll{Z_{(i_2)}} C_3 \cdots C_n \Rsyll{Z_{(i_n)}} C_{n+1}\\
  &=& D_1 \Rsyll{Z_{(i_1)}} D_1\inv D_2 \Rsyll{Z_{(i_2)}} D_2\inv
  D_3\cdots D_{n-1}\inv D_n\Rsyll{Z_{(i_n)}} D_n\inv D_{n+1}\\
  &=& \Rsyll{D_1Z_{(i_1)}D_1\inv} \Rsyll{D_2Z_{(i_2)}D_2\inv}
  \cdots \Rsyll{D_n Z_{(i_n)} D_n\inv} D_{n+1} \\
  &=& \Rsyll{P_1} \Rsyll{P_2} \cdots \Rsyll{P_n} D_{n+1},
\end{eqnarray*}
where $P_k = D_kZ_{(i_k)}D_k\inv$. Therefore, every Clifford+$T$
operator can be written as a product of Pauli rotations followed by a
single Clifford operator, as claimed. It also shows that the number of
required Pauli rotations is at most equal to the $T$-count of the
original circuit. In fact, since every Pauli rotation has $T$-count 1,
it is clear that every product of $n$ Pauli rotations can be converted
to a circuit of $T$-count $n$, and vice versa. In particular, the
minimal $T$-count of a circuit is equal to the minimal number of Pauli
rotations required to express it.

The Pauli rotation representation is not unique. There are some
obvious relations:
\begin{itemize}
\item[(a)] $\Rsyll{P}$ and $\Rsyll{Q}$ commute if and only if $P$ and $Q$
  commute. This follows from {\eqref{eqn:def-rsyll}}.
\item[(b)] For any $P$, the operator $\Rsyll{P}^2$ is
  Clifford, and therefore can be eliminated, resulting in a shorter
  word.  To see why, recall that there exists a Clifford
  operator $C$ such that $\Rsyll{P}=CT_{(1)}C\inv$; therefore
  $\Rsyll{P}^2=CT_{(1)}^2C\inv$. Since $T_{(1)}^2=S_{(1)}$ is a
  Clifford gate, it follows that $\Rsyll{P}^2$ is Clifford.
\item[(c)] For any $P$, there exists a Clifford operator $D$ such that
  $\Rsyll{(-P)} = \Rsyll{P}D$. Indeed, let $C$ be a Clifford operator
  such that $P = CZ_{(1)}C\inv$. Then $-P=C(-Z_{(1)})C\inv =
  CX_{(1)}Z_{(1)}X_{(1)}C\inv$. Therefore
  $\Rsyll{(-P)}=CX_{(1)}T_{(1)}X_{(1)}C\inv$. Using the relation $XTX =
  T S^\dagger \omega$, we have
  $\Rsyll{(-P)}=CT_{(1)}S_{(1)}^\dagger\omega C\inv = CT_{(1)}C\inv
  CS_{(1)}^\dagger\omega C\inv = \Rsyll{P}CS_{(1)}^\dagger\omega
  C\inv$. Thus, the claim holds with $D=CS_{(1)}^\dagger\omega C\inv$.
\end{itemize}
It is relatively easy to standardize the Pauli rotation representation
modulo the above three relations: First, we eliminate any generators
of the form $\Rsyll{(-P)}$. This can be done from left to right, using
relations from (c); the resulting Clifford operator can be shifted all
the way to the end of the word using relations of the form
$D R_P = R_{Q}D$, where $Q=DPD\inv$, see
{\eqref{eqn:conjugate}}. Next, we use relations from (a) to swap
adjacent generators when possible, for example arriving at the
lexicographically smallest word that is equal to the given word up to
such commuting permutations. Next, we use relations from (b) to remove
any duplicates. Should there be any such duplicates, the resulting
word will need to be standardized again, but since it uses fewer Pauli
rotations, the process eventually terminates.

However, even when the Pauli rotation representation is standardized
modulo the relations (a), (b), and (c), it is still not
unique. Indeed, there are some ``non-obvious'' relations. In a
sense, the contribution of this paper is to state exactly what these
non-obvious relations are. They turn out to be the following. Here,
for brevity, we have omitted the tensor symbol $\otimes$, i.e., we
wrote $\Rsyll{IX}$ instead of $\Rsyll{I\otimes X}$.
\[
  \begin{array}{rcl}
    \Rsyll{IX}\Rsyll{IZ}\Rsyll{ZZ}\Rsyll{ZX} &=& 
    \Rsyll{ZX}\Rsyll{IZ}\Rsyll{ZZ}\Rsyll{IX}, \\
    \Rsyll{IX}\Rsyll{IZ}\Rsyll{IX}\Rsyll{ZX}\Rsyll{ZZ}\Rsyll{ZX} &=& 
    \Rsyll{ZX}\Rsyll{IZ}\Rsyll{IX}\Rsyll{ZX}\Rsyll{ZZ}\Rsyll{IX}, \\
    \Rsyll{XY}\Rsyll{YZ}\Rsyll{XZ}\Rsyll{IX}\Rsyll{ZI}
    \Rsyll{YX}\Rsyll{ZY}\Rsyll{ZX}\Rsyll{XI}\Rsyll{IZ} &=& 
    \Rsyll{YX}\Rsyll{ZY}\Rsyll{ZX}\Rsyll{XI}\Rsyll{IZ}
    \Rsyll{XY}\Rsyll{YZ}\Rsyll{XZ}\Rsyll{IX}\Rsyll{ZI}.
  \end{array}
\]
These turn out to be equivalent to relations {\eqref{eqn-rela}},
{\eqref{eqn-relb}}, and {\eqref{eqn-relc}} in
Figure~\ref{fig-relations}, respectively.  We will address the
question of what these relations might ``mean'' (i.e., how one might
be able to see that they are true without computing the matrices) in
Section~\ref{sec:axioms}.

% ----------------------------------------------------------------------
\subsection{Proof assistants}

As outlined in Section~\ref{sec:outline}, once we are armed with the
Reidemeister-Schreier theorem, in theory there is a mechanical way to
obtain a complete set of relations for $\CliffordT(2)$, given that
$\CliffordT(2)$ is a subgroup of $U_4(\Z[\frac{1}{\sqrt{2}},i])$ and
we already have a complete set of relations for the latter due to
Greylyn {\cite{greylyn2014generators}}. However, when applied in
practice, this method yields a large number of very large relations,
all of which must be shown to follow from the relations in
Figure~\ref{fig-relations}. Although
Figure~\ref{fig:greylyn-relations} appears to contain only 20
relations, they are actually parameterized by indices such as
$\jay,\kay$, etc. After accounting for these indices, there are 123
distinct relations. Since there are two cosets of $\CliffordT(2)$ in
$U_4(\Z[\frac{1}{\sqrt{2}},i])$, under part (b) of the
Reidemeister-Schreier theorem, each of these 123 relations yields two
Clifford+$T$ relations, plus another 8 relations (one for each
generator) from part (a), giving a total of 254 Clifford+$T$ relations
that must be verified.  This task is too daunting to do ``by hand''.

Given the mechanical and repetitive nature of these calculations, we
initially wrote a computer program to generate and verify the
relations. However, this raised another issue: our program was large
and complicated and used a variety of tactics to show that the given
relations follow from the ones in Figure~\ref{fig-relations}. We could
not claim with mathematical certainty that our program was free of
bugs, nor that it didn't use some hidden assumptions that weren't
actually consequences of Figure~\ref{fig-relations}. Moreover, it
would have been unreasonable for any referee to verify our calculations.

For this reason, we decided to go one step further and formalize the
soundness and completeness proofs in a {\em proof assistant}. A proof
assistant is a piece of software in which one can write definitions,
theorems, and proofs, and the software will check the correctness of
the proofs. Purists might object that the proof assistant is itself a
piece of software that might be buggy. But, as has been argued
eloquently by {\cite{Gonthier2008,Hales2008}}, current proof
assistants can be scrutinized at many levels and are many orders of
magnitude more reliable than the traditional way of checking
paper-and-pencil proofs. The particular proof assistant we used in
this work is Agda {\cite{Agda}}.

% ----------------------------------------------------------------------
\section{Proof of the main result}
\label{sec:proof}
% ----------------------------------------------------------------------
\subsection{Soundness and completeness}
\label{ssec:soundness-and-completeness}

Our goal is to prove that Theorem~\ref{thm:greylyn} implies
Theorem~\ref{thm:main}. Recall that Greylyn's set of generators for
$U_4(R)$ is
$\Yy=\s{\omega_{[j]}, X_{[j,k]}, H_{[j,k]} \mid
  \mbox{$j,k\in\s{1,...,4}$ and $j<k$}}$. Also recall that our target
set of generators for $\CliffordT(2)$ is
$\Xx=\s{\omega,H_0,H_1,S_0,S_1,T_0,T_1,\CZ}$. We fix a translation
from $\Xx$ to $\Yy^*$ as follows:
\[
\def\arraystretch{1.2}
\begin{array}{lcl}
  f(\omega) &=& \omega_{[0]}\omega_{[1]}\omega_{[2]}\omega_{[3]}, \\
  f(H_0)    &=& H_{[1,3]}H_{[0,2]}, \\
  f(H_1)    &=& H_{[2,3]}H_{[0,1]}, \\
  f(S_0)    &=& \omega_{[2]}^2\omega_{[3]}^2, \\
  f(S_1)    &=& \omega_{[1]}^2\omega_{[3]}^2, \\
  f(T_0)    &=& \omega_{[2]}\omega_{[3]}, \\
  f(T_1)    &=& \omega_{[1]}\omega_{[3]}, \\
  f(\CZ)    &=& \omega_{[3]}^4. \\
\end{array}
\]
We prove the following soundness and completeness theorems for this
translation:
\begin{theorem}[Soundness]\label{thm:soundness}
  For all $w,v\in\Xx^*$, $w \simS v$ implies $f^*(w)\simR f^*(v)$.
\end{theorem}

\begin{theorem}[Completeness]\label{thm:completeness}
  For all $w,v\in\Xx^*$, $f^*(w)\simR f^*(v)$ implies $w \simS v$.
\end{theorem}

As already noted in Section~\ref{ssec:reidemeister-schreier}, these
two theorems, together with Theorem~\ref{thm:greylyn}, immediately
imply Theorem~\ref{thm:main}. Specifically, we have $w\simS v$ if and
only if $f^*(w)\simR f^*(v)$ if and only if
$\sem{f^*(w)} = \sem{f^*(v)}$ if and only if $\sem{w} = \sem{v}$,
where the first equivalence follows from Theorems~\ref{thm:soundness}
and {\ref{thm:completeness}}, the second equivalence follows from
Theorem~\ref{thm:greylyn}, and the last equivalence holds because the
function $f$ respects the interpretation.

% ----------------------------------------------------------------------
\subsection{The formal proof}

Soundness and completeness are formally proved in the Agda code
accompanying this paper {\cite{Agda-code}}. We organized the code to
make it hopefully as easy as possible to verify the result. The code
consists of 67 files that are listed in Figure~\ref{fig:files}, and
which we now briefly describe.

\paragraph{(a) Background.}

The eight files in the ``background'' section contain general-purpose
definitions of the kind that are usually found in the Agda standard
library, i.e., basic properties of booleans, integers, equality,
propositional connectives, etc. The reason we did not use the actual
Agda standard library is that it is very large and changes frequently.
We felt that it is better for our code to be self-contained rather
than depending on a particular library version.

\paragraph{(b) Statement of the result.}

In these two files, we give a minimal set of definitions that allows
us to {\em state} the soundness and completeness theorems. The file
{\tt Word.agda} defines what it means to be a word over a set of
generators, as well as the inference rules we use for deriving
relations from a set of axioms (such as reflexivity, symmetry,
transitivity, congruence, associativity, and the left and right unit
laws). Note that in the Agda code, we define a word as a term in the
language of monoids, rather than as a sequence of generators. In other
words, associativity and the unit laws are treated as laws, rather
than being built into the definition.  The file {\tt Word.agda} also
defines the $f^*$ operation used in the statement of the soundness and
completeness theorems.  The file {\tt Generator.agda} defines the
Clifford+$T$ generators and the relations from
Figure~\ref{fig-relations}, Greylyn's generators and the relations
from Figure~\ref{fig:greylyn-relations}, and the translation function
$f$ from Section~\ref{ssec:soundness-and-completeness}. It also
contains the statement of the soundness and completeness theorems, but
not their proofs. The reason we state these theorems separately from
their proofs is to make sure that Agda (and a human reviewer) can
verify that the statement of these theorems only depends on the
relatively small number of definitions given so far, and not on the
much larger number of definitions and tactics used in the proof.

\paragraph{(c) Details of the proof.}

The proof of the soundness and completeness theorems relies on a large
number of auxiliary definitions and lemmas, and comprises the bulk of
our code with 56 files.  This includes a formal proof of the
Reidemeister-Schreier theorem; several tactics for automating steps in
certain equational proofs; a simplified presentation of Greylyn's
generators and relations, using only 5 generators and 19 relations
(instead of Greylyn's original 16 generators and 123 relations), along
with the proof of its completeness; a formalization of Pauli rotations
and their relevant properties; as well as 46 step-by-step proofs of
individual relations. These details are primarily intended to be
machine-readable, and can safely be skipped by readers who trust
Agda and merely want to check the proof rather than reading
it. However, all of the files are documented and human-readable.

The relations in the files {\tt Equation1.agda} to {\tt
  Equation46.agda} are at the heart of the completeness proof. These
are the relations that must be proved to satisfy the hypotheses of the
Reidemeister-Schreier theorem. Some of these relations are trivial,
such as {\tt Equation13.agda}. Others are highly non-trivial and
require almost a thousand proof steps, such as {\tt Equation44.agda}.
In particular, the proofs that require relations
{\eqref{eqn-rela}}--{\eqref{eqn-relc}} from Figure~\ref{fig-relations}
tend to be non-obvious; in fact, this is how we discovered relations
{\eqref{eqn-rela}}--{\eqref{eqn-relc}} in the first place. We did not
write these equational proofs by hand; instead, we used a
semi-automated process where most of the proofs were generated by a
separate Haskell program and output in a format that is convenient and
efficient for Agda to check. Originally, we also attempted to write
Agda tactics that would allow Agda to derive these relations fully
automatically; however, this failed due to performance issues with
Agda.

\paragraph{(d) Proof witness.}

Finally, the file {\tt Proof.agda} contains nothing but a witness of
the fact that the soundness and completeness theorems have been
formally proven. A reader who wants to skip the details of the formal
proof only needs to check two things: the statement of the main result
in {\tt Generator.agda} (to make sure the statement correctly captures
what we said it does), and the fact that the Agda proof checker
accepts {\tt Proof.agda}.

% ......................................................................
\begin{figure}
  \def\fn#1{{\tt\footnotesize #1}}
  \def\ex#1{{\small #1}}
  
  (a) Background:

  \smallskip
  
  \begin{tabular}{p{2.08in}@{~}p{4.1in}}
    \fn{Boolean.agda} & \ex{The type of booleans.} \\
    \fn{Proposition.agda} & \ex{Basic definitions in propositional logic.} \\
    \fn{Equality.agda} & \ex{Basic properties of equality.} \\
    \fn{Decidable.agda} & \ex{Some definitions to deal with decidable properties.} \\
    \fn{Inspect.agda} & \ex{Agda's ``inspect'' paradigm, to assist with pattern matching.} \\
    \fn{Nat.agda} & \ex{Basic properties of the natural numbers.} \\
    \fn{Maybe.agda} & \ex{The ``Maybe'' type.} \\
    \fn{List.agda} & \ex{Basic properties of lists.} \\
  \end{tabular}
  
  \bigskip

  (b) Statement of the result

  \smallskip

  \begin{tabular}{p{2.08in}@{~}p{4.1in}}
    \fn{Word.agda} & \ex{Basic properties of words.} \\
    \fn{Generator.agda} & \ex{Generators and relations for our two
    groups, and statement of main result.} \\
  \end{tabular}
  
  \bigskip

  (c) Proof of the result
  
  \smallskip

  \begin{tabular}{p{2.08in}@{~}p{4.1in}}
    \fn{Word-Lemmas.agda} & \ex{Basic lemmas about monoids and groups, and equational reasoning.} \\
    \fn{Reidemeister-Schreier.agda} & \ex{Two versions of the Reidemeister-Schreier theorem.} \\
    \fn{Word-Tactics.agda} & \ex{Some tactics for proving properties of words.} \\
    \fn{Clifford-Lemmas.agda} & \ex{A decision procedure for equality of 2-qubit Clifford operators.} \\
    \fn{CliffordT-Lemmas.agda} & \ex{Properties and tactics for Clifford+$T$ operators.} \\
    \fn{Greylyn-Lemmas.agda} & \ex{Some automation for Greylyn's 1- and 2-level operators.} \\
    \fn{Soundness.agda} & \ex{Proof of soundness.} \\
    
    \fn{Greylyn-Simplified.agda} & \ex{A smaller set of generators and relations for Greylyn's operators.} \\
    \fn{PauliRotations.agda} & \ex{Definitions, properties, and tactics for Pauli rotations.} \\
    \fn{Equation1.agda} -- \fn{Equation46.agda} & \ex{Explicit proofs of 46 relations required for completeness.} \\
    \fn{Completeness.agda} & \ex{Proof of completeness.} \\
  \end{tabular}
  
  \bigskip

  (d) Top-level proof witness
  
  \smallskip

  \begin{tabular}{p{2.08in}@{~}p{4.1in}}
    \fn{Proof.agda} & \ex{The final witness for soundness and completeness.} \\
  \end{tabular}

  \caption{List of Agda files. The files are listed in order of
    dependency, i.e., each file only imports earlier files.}\label{fig:files}
\end{figure}
% ......................................................................

% ----------------------------------------------------------------------
\section{Discussion of the axioms}
\label{sec:axioms}

Here, we give some further perspectives on what the axioms of
Figure~\ref{fig-relations} might ``mean'', and in particular, how one
might convince oneself that the relations are true without having to
compute the corresponding matrices.

Note that we are not claiming that axioms
{\eqref{rel-c1}}--{\eqref{rel-c20}} are independent; for example,
{\eqref{rel-c8}} clearly follows from {\eqref{rel-c14}} and
{\eqref{rel-c16}}; however, we found it useful to separate the
Clifford relations from the rest, which is why {\eqref{rel-c8}} was
included. It would be nice to know whether axioms
{\eqref{eqn-rela}}--{\eqref{eqn-relc}} are independent from the others
and from each other, and this seems likely to be true, but we do not
know.

The axioms in groups (a)--(c) are well-known; they merely express the
Clifford relations {\cite{selinger2013generators}} and the fact that
operators on disjoint qubits commute.  Relations~{\eqref{rel-c14}} and
{\eqref{rel-c15}} express the well-known facts that $T^2=S$ and
$(TX)^2=\omega$, whereas relation~{\eqref{rel-c16}} holds because
diagonal operators commute. Note that the upside-down version of
relation~{\eqref{rel-c16}} was not included among our axioms; this is
because it is actually derivable from the remaining axioms.
Relation~{\eqref{rel-c17}} becomes obvious once one realizes that the
swap gate can be expressed as a sequence of three controlled-not
gates:
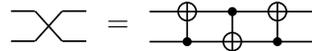
\[
\m{\begin{qcircuit}[scale=0.5]
    \grid{2.5}{0,1}
    \swapgate{1.25}{0}{1};
  \end{qcircuit}
}
~=~
\m{\begin{qcircuit}[scale=0.5]
    \grid{5.5}{0,1}
    \controlled{\notgate}{1.25,1}{0};
    \controlled{\notgate}{2.75,0}{1};
    \controlled{\notgate}{4.25,1}{0};
  \end{qcircuit}
}
\]
Relation~{\eqref{rel-c17}} is then obtained by simplifying the
following, which expresses the fact that a $T$-gate can be moved past
a swap-gate:
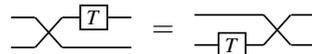
\[
\m{\begin{qcircuit}[scale=0.5]
    \grid{4}{0,1}
    \swapgate{1.25}{0}{1};
    \gate{$T$}{2.75,1};
  \end{qcircuit}
}
~=~
\m{\begin{qcircuit}[scale=0.5]
    \grid{4}{0,1}
    \gate{$T$}{1.25,0};
    \swapgate{2.75}{0}{1};
  \end{qcircuit}
}
\]

We will now focus on the ``non-obvious'' relations
{\eqref{eqn-rela}}--{\eqref{eqn-relc}}. Relations~{\eqref{eqn-rela}}
and {\eqref{eqn-relb}} are of the form
\begin{equation}\label{eqn-4}
    \scalebox{1}{\m{
 	\begin{qcircuit}[scale=0.5]
          \grid{6.00}{0,1}
          \controlled{\notgate}{1.00,0}{1};
          \gate{$A$}{2.25,0};
          \whitecontrolled{\notgate}{3.5,0}{1};
          \gate{$A^\dagger$}{4.75,0};
        \end{qcircuit}
    }}
    ~=~
    \scalebox{1}{\m{
 	\begin{qcircuit}[scale=0.5]
          \grid{6}{0,1}
          \gate{$A$}{1.25,0};
          \whitecontrolled{\notgate}{2.5,0}{1};
          \gate{$A^\dagger$}{3.75,0};
          \controlled{\notgate}{5,0}{1};
        \end{qcircuit}
    }}.
\end{equation}
They hold because positively controlled gates commute with negatively
controlled gates. Note that there are infinitely many relations of the
form {\eqref{eqn-4}}, where $A$ is any single-qubit Clifford+$T$
operator, but our completeness proof shows that, in the presence of
the remaining axioms, two of them are sufficient to prove all the
others.

Relation~{\eqref{eqn-relc}} is more interesting. It, too, states that
two operators commute, but it is less obvious why this is so. Ideally,
we would be able to find some simpler and more obvious relations that
imply {\eqref{eqn-relc}}. While we have not been able to find
such simpler relations in the Clifford+$T$ generators, we can do this
if we permit ourselves a controlled $T$-gate. Note that the controlled
$T$-gate is not itself a member of the 2-qubit Clifford+$T$ group,
since representing it as a Clifford+$T$ operator requires an ancilla
{\cite{GilesSelinger2013}}. But the use of controlled $T$-gates is
nevertheless helpful in explaining relation~{\eqref{eqn-relc}}. We
start by noting that the controlled $T$-gate satisfies the following
obvious circuit identities (and their upside-down versions):
\begin{eqnarray}
    \m{\begin{qcircuit}[scale=0.5]
        \grid{2.5}{0,1}
        \controlled{\gate{$T$}}{1.25,0}{1}
      \end{qcircuit}
    } 
    &=& 
    \m{\begin{qcircuit}[scale=0.5]
        \grid{2.5}{0,1}
        \controlled{\gate{$T$}}{1.25,1}{0}
      \end{qcircuit}
    }
    \label{eqn-5}
    \\\nonumber\\[0ex]
    \m{\begin{qcircuit}[scale=0.5]
        \grid{3.75}{0,1}
        \whitecontrolled{\gate{$T$}}{1.25,0}{1}
        \controlled{\gate{$T$}}{2.5,0}{1}
      \end{qcircuit}
    } 
    &=& 
    \m{\begin{qcircuit}[scale=0.5]
        \grid{2.5}{0,1}
        \gate{$T$}{1.25,0}
      \end{qcircuit}
    } 
    \label{eqn-6}
    \\\nonumber\\[0ex]
    \m{\begin{qcircuit}[scale=0.5]
        \grid{3.75}{0,1}
        \gate{$T$}{1.25,1}
        \whitecontrolled{\gate{$T$}}{2.5,0}{1}
      \end{qcircuit}
    } 
    &=& 
    \m{\begin{qcircuit}[scale=0.5]
        \grid{3.75}{0,1}
        \gate{$T$}{1.25,0}
        \whitecontrolled{\gate{$T$}}{2.5,1}{0}
      \end{qcircuit}
    }
    \label{eqn-7}
    \\\nonumber\\[0ex]
    \m{\begin{qcircuit}[scale=0.5]
        \grid{5}{0,1}
        \whitecontrolled{\gate{$H$}}{1.25,0}{1}
        \gate{$T$}{2.5,0}
        \controlled{\gate{$H$}}{3.75,0}{1}
      \end{qcircuit}
    } 
    &=& 
    \m{\begin{qcircuit}[scale=0.5]
        \grid{5}{0,1}
        \controlled{\gate{$T$}}{1.25,0}{1}
        \gate{$H$}{2.5,0}
        \whitecontrolled{\gate{$T$}}{3.75,0}{1}
      \end{qcircuit}
    }.
    \label{eqn-8}
\end{eqnarray}
Identities {\eqref{eqn-5}}--{\eqref{eqn-7}} are obvious because all of
the operators in them are diagonal. Identity {\eqref{eqn-8}} holds by
case distinction: this circuit applies either $HT$ or $TH$ to the
bottom qubit, depending on whether the top qubit is $\ket{0}$ or
$\ket{1}$. Using these identities, we can easily prove
{\eqref{eqn-relc}}:
\[
\begin{array}{rcl}
    \scalebox{1}{\m{
        \begin{qcircuit}[scale=0.5]
          \grid{10.00}{0,1}
          \whitecontrolled{\gate{$H$}}{1.25,0}{1};
          \gate{$T$}{2.75,0};
          \controlled{\gate{$H$}}{4.25,0}{1};
          \whitecontrolled{\gate{$H$}}{5.75,1}{0};
          \gate{$T$}{7.25,1};
          \controlled{\gate{$H$}}{8.75,1}{0};
        \end{qcircuit}
    }} &\stackrel{\eqref{eqn-8}}{=}& \scalebox{1}{\m{
        \begin{qcircuit}[scale=0.5]
          \grid{10.00}{0,1}
          \controlled{\gate{$T$}}{1.25,0}{1};
          \gate{$H$}{2.75,0};
          \whitecontrolled{\gate{$T$}}{4.25,0}{1};
          \controlled{\gate{$T$}}{5.75,1}{0};
          \gate{$H$}{7.25,1};
          \whitecontrolled{\gate{$T$}}{8.75,1}{0};
        \end{qcircuit}
    }} \\\\[-1.5ex]&\stackrel{\eqref{eqn-5}}{=}& \scalebox{1}{\m{
        \begin{qcircuit}[scale=0.5]
          \grid{10.00}{0,1}
          \controlled{\gate{$T$}}{1.25,0}{1};
          \gate{$H$}{2.75,0};
          \whitecontrolled{\gate{$T$}}{4.25,0}{1};
          \controlled{\gate{$T$}}{5.75,0}{1};
          \gate{$H$}{7.25,1};
          \whitecontrolled{\gate{$T$}}{8.75,1}{0};
        \end{qcircuit}
    }} \\\\[-1.5ex]&\stackrel{\eqref{eqn-6}}{=}& \scalebox{1}{\m{
        \begin{qcircuit}[scale=0.5]
          \grid{10.00}{0,1}
          \controlled{\gate{$T$}}{1.25,0}{1};
          \gate{$H$}{2.75,0};
          \gate{$T$}{5,0};
          \gate{$H$}{7.25,1};
          \whitecontrolled{\gate{$T$}}{8.75,1}{0};
        \end{qcircuit}
    }} \\\\[-1.5ex]&=& \scalebox{1}{\m{
        \begin{qcircuit}[scale=0.5]
          \grid{10.00}{0,1}
          \controlled{\gate{$T$}}{1.25,0}{1};
          \gate{$H$}{4.25,0};
          \gate{$T$}{7.25,0};
          \gate{$H$}{4.25,1};
          \whitecontrolled{\gate{$T$}}{8.75,1}{0};
        \end{qcircuit}
    }} \\\\[-1.5ex]&\stackrel{\eqref{eqn-5},\eqref{eqn-7}}{=}& \scalebox{1}{\m{
        \begin{qcircuit}[scale=0.5]
          \grid{10.00}{0,1}
          \controlled{\gate{$T$}}{1.25,1}{0};
          \gate{$H$}{4.25,1};
          \gate{$T$}{7.25,1};
          \gate{$H$}{4.25,0};
          \whitecontrolled{\gate{$T$}}{8.75,0}{1};
        \end{qcircuit}
    }} \\\\[-1.5ex]&=& \scalebox{1}{\m{
        \begin{qcircuit}[scale=0.5]
          \grid{10.00}{0,1}
          \controlled{\gate{$T$}}{1.25,1}{0};
          \gate{$H$}{2.75,1};
          \gate{$T$}{5,1};
          \gate{$H$}{7.25,0};
          \whitecontrolled{\gate{$T$}}{8.75,0}{1};
        \end{qcircuit}
    }} \\\\[-1.5ex]&\stackrel{\eqref{eqn-6}}{=}& \scalebox{1}{\m{
        \begin{qcircuit}[scale=0.5]
          \grid{10.00}{0,1}
          \controlled{\gate{$T$}}{1.25,1}{0};
          \gate{$H$}{2.75,1};
          \gate{$H$}{7.255,0};
          \whitecontrolled{\gate{$T$}}{4.25,1}{0};
          \controlled{\gate{$T$}}{5.75,1}{0};
          \whitecontrolled{\gate{$T$}}{8.75,0}{1};
        \end{qcircuit}
    }} \\\\[-1.5ex]&\stackrel{\eqref{eqn-5}}{=}& \scalebox{1}{\m{
        \begin{qcircuit}[scale=0.5]
          \grid{10.00}{0,1}
          \controlled{\gate{$T$}}{1.25,1}{0};
          \gate{$H$}{2.75,1};
          \gate{$H$}{7.255,0};
          \whitecontrolled{\gate{$T$}}{4.25,1}{0};
          \controlled{\gate{$T$}}{5.75,0}{1};
          \whitecontrolled{\gate{$T$}}{8.75,0}{1};
        \end{qcircuit}
    }} \\\\[-1.5ex] &\stackrel{\eqref{eqn-8}}{=}& \scalebox{1}{\m{
        \begin{qcircuit}[scale=0.5]
          \grid{10.00}{0,1}
          \whitecontrolled{\gate{$H$}}{1.25,1}{0};
          \gate{$T$}{2.75,1};
          \controlled{\gate{$H$}}{4.25,1}{0};
          \whitecontrolled{\gate{$H$}}{5.75,0}{1};
          \gate{$T$}{7.25,0};
          \controlled{\gate{$H$}}{8.75,0}{1};
        \end{qcircuit}
    }}.
\end{array}
\]
Note that there is again an infinite family of such relations, because
in the above derivation, we could have used any gate in place of $H$.
However, due to completeness, all other such relations are
consequences of {\eqref{eqn-rela}}--{\eqref{eqn-relc}} and the
remaining axioms.

Another way to look at
relations~{\eqref{eqn-rela}}--{\eqref{eqn-relc}} is in terms of their
Pauli rotation representations. As we already mentioned in
Section~\ref{ssec:pauli}, up to basis changes, the three relations can be
written in terms of Pauli rotations, respectively as follows:
\[
  \begin{array}{rcl}
    \Rsyll{IX}\Rsyll{IZ}\Rsyll{ZZ}\Rsyll{ZX} &=& 
    \Rsyll{ZX}\Rsyll{IZ}\Rsyll{ZZ}\Rsyll{IX}, \\
    \Rsyll{IX}\Rsyll{IZ}\Rsyll{IX}\Rsyll{ZX}\Rsyll{ZZ}\Rsyll{ZX} &=& 
    \Rsyll{ZX}\Rsyll{IZ}\Rsyll{IX}\Rsyll{ZX}\Rsyll{ZZ}\Rsyll{IX}, \\
    \Rsyll{XY}\Rsyll{YZ}\Rsyll{XZ}\Rsyll{IX}\Rsyll{ZI}
    \Rsyll{YX}\Rsyll{ZY}\Rsyll{ZX}\Rsyll{XI}\Rsyll{IZ} &=& 
    \Rsyll{YX}\Rsyll{ZY}\Rsyll{ZX}\Rsyll{XI}\Rsyll{IZ}
    \Rsyll{XY}\Rsyll{YZ}\Rsyll{XZ}\Rsyll{IX}\Rsyll{ZI}.
  \end{array}
\]
When written in this form, the first two of these relations only use
$X$ and $Z$ Paulis, and use only $Z$ on the left qubit. This indicates
that these relations are about controlled gates. We can also see that
in both cases, the relation exchanges the positions of the leftmost
$\Rsyll{IX}$ and the rightmost $\Rsyll{ZX}$. The first relation can
also be seen to express the fact that $\Rsyll{IZ}\Rsyll{ZZ}$ commutes
with $\Rsyll{ZX}\Rsyll{IX}\inv$, and similarly for the second
relation.  The third relation again takes the form of an operator
commuting with its upside-down version.

% ----------------------------------------------------------------------
\section{Conclusion and future work}
\label{sec:conclusion}
We gave a presentation of the 2-qubit Clifford+$T$ group by generators
and relations. We did this by applying the Reidemeister-Schreier
theorem to Greylyn's presentation of the group of unitary
$4\times 4$-matrices over the ring $\Z[\frac{1}{\sqrt{2}},i]$. Since
there is a very large number of relations to check and simplify, and
checking them by hand or by an unverified computer program would be
error-prone, we used the proof assistant Agda to formalize our
proof. The latter process is painstaking and took us more than 5 years
to complete after our result was first announced in {\cite{BS2015}}.

An obvious candidate for future work would be to find a complete set
of relations for the Clifford+$T$ group with 3 or more qubits. This is
currently out of reach for two reasons: first, the computations
required to simplify any potential set of relations will be even more
labor-intensive than in the 2-qubit case. Second, and more seriously,
there is no known presentation of the group of unitary
$n\times n$-matrices over the ring $\Z[\frac{1}{\sqrt{2}},i]$ for
$n>4$.

Another project that is currently in progress is to apply the method
of this paper to restrictions of the Clifford+$T$ group for which
presentations of the corresponding matrix group are known. This
includes the Clifford+Toffoli gate set and the Clifford+controlled-$S$
gate set.

%----------------------------------------------------------------------
\bibliographystyle{eptcs}
\bibliography{cliffordt2}

%----------------------------------------------------------------------
\end{document}